\documentclass[letterpaper, 10 pt, conference]{ieeeconf}  

\IEEEoverridecommandlockouts                              
\usepackage{graphics} 
\usepackage{epsfig} 
\usepackage{times} 
\usepackage{amsmath} 
\usepackage{amssymb}  
\usepackage[tight,footnotesize]{subfigure}
\usepackage{enumerate}
\usepackage{cite}
\DeclareMathOperator{\atantwo}{atan2}
\newtheorem{theorem}{Theorem}
\newtheorem{lemma}{Lemma}

\usepackage[utf8]{inputenc}
\usepackage[english]{babel}

\title{\LARGE \bf
Modelling and Dynamic Tracking Control of Industrial Vehicles with Tractor-trailer Structure
}

\author{Hongchao~Zhao, Zhe~Liu, Zhiqiang~Li, Shunbo~Zhou, Wen~Chen, Chuanzhe~Suo, and Yun-Hui~Liu
\thanks{This work is supported in part by the NSFC under Grant U1613218, in part by the Hong Kong ITC under Grant ITS/448/16FP, in part by the National Key Research and Development Program of China under Grant 2018YFB1309300, and in part by the VC Fund 4930745 of the CUHK T Stone Robotics Institute.}%
\thanks{All the authors are with the T Stone Robotics Institute and Department of Mechanical and Automation Engineering, The Chinese University of Hong Kong, Shatin, HKSAR. Corresponding authors: Z. Liu and H. Zhao (zheliu@cuhk.edu.hk; hongchao.zhao@link.cuhk.edu.hk).}%
}

\begin{document}

\maketitle
\thispagestyle{empty}
\pagestyle{empty}

\begin{abstract}
Existing works on control of tractor-trailers systems only consider the kinematics model without taking dynamics into account. Also, most of them treat the issue as a pure control theory problem whose solutions are difficult to implement. This paper presents a trajectory tracking control approach for a full-scale industrial tractor-trailers vehicle composed of a car-like tractor and arbitrary number of passive full trailers. To deal with dynamic effects of trailing units, a force sensor is innovatively installed at the connection between the tractor and the first trailer to measure the forces acting on the tractor. The tractor's dynamic model that explicitly accounts for the measured forces is derived. A tracking controller that compensates the pulling/pushing forces in real time and simultaneously drives the system onto desired trajectories is proposed. The propulsion map between throttle opening and the propulsion force is proposed to be modeled with a fifth-order polynomial. The parameters are estimated by fitting experimental data, in order to provide accurate driving force. Stability of the control algorithm is rigorously proved by Lyapunov methods. Experiments of full-size vehicles are conducted to validate the performance of the control approach.
\end{abstract}

\section{INTRODUCTION}
Cargo transportation is one of the major tasks in big warehouses, railway stations, cargo terminals, airports, etc. A tractor pulling multiple passive trailers has been widely utilized for such tasks because of its higher efficiency and lower cost, compared to a group of individual vehicles. An autonomous tractor-trailers vehicle can further increase the productivity and give a perfect solution to the serious shortage of human drivers and their rapidly growing wages. However, these systems with payloads on trailers pose a complex nonlinear and underactuated dynamic control problem.

Numerous works in the literature have been devoted to this problem. Different control methods are proposed depending on the type of trailer hitching. A hitching is called ``on-axle" when the rotary hitching joint is located at middle point of the trailer's wheel axles (points marked A in Fig. \ref{realandgeomodel}). While a hitching is called ``off-axle" when the hitching joint is located somewhere between two consecutive trailers (points marked B in Fig. \ref{realandgeomodel}). Off-axle Kingpin hitching is broadly deployed in real applications thanks to its simpler mechanics, though the control problem is more difficult than that of on-axle joints. In our case, both on-axle and off-axle joints exist in the vehicle system (Fig.
\ref{realandgeomodel}).

Early contributions to tractor-trailers vehicle control mainly focused on on-axle hitching systems using nonlinear control theory. Laumond \cite{1} proved the controllability of on-axle n-trailer vehicles with tools from differential geometry. Tilbury \textit{et al.} \cite{2} proposed a Goursat canonical form, which is dual to the chained form, to model a class of driftless non-holonomic systems including the tractor with N trailers systems. Sordalen and Wichlund \cite{3} then achieved the exponential stability using exact feedback linearization with a time-varying state feedback control law, which provided a theoretical basis for stabilizing the system at a goal pose or a path with constant curvature. Lamiraux and Laumond \cite{4} applied the method for global obstacle-free path planning and trajectory following, but the results were quite complex. On the other side, however, these feedback linearization techniques cannot be applied for systems with two or more off-axle trailers, since the kinematics loses the property of differential flatness and becomes not feedback linearizable (Rouchon \textit{et al.} \cite{5}, Bushnell \cite{6}).  This is why Lizarraga \textit{et al.} \cite{7} utilized some configurations that are able to be approximated locally by a chained form such that exponential stabilization can be achieved. Bolzern \textit{et al.} \cite{8} approximated the kinematics of off-axle connection by an on-axle system model which has similar steady response and model differences are regarded as disturbances. Alternatively, control of off-axle tractor-trailers systems has also been approached by means of off-tracking analysis, i.e., analysis of trailer's tracking deviation from the tractor's path. This is mainly due to the property of open-loop stability. Different from the on-axle case, the trailers' tracking errors do not necessarily grow with increasing number of units (Lee \textit{et al.} \cite{9}). In this sense, Bushnell \textit{et al.} \cite{10} investigated the off-tracking bounds in transition from a beeline to an arc of a circle and vice versa. This enables conventional obstacle-avoidance path planners to be applied as if the tractor-trailers system is an enlarged tractor \cite{11}. It is easy to observe that existing works on control of tractor-trailers systems only consider the kinematics model without taking dynamic effects into account. Besides, most of them treat the issue as a pure control theory problem whose solutions are difficult to implement.

\begin{figure}[!t]
	\centering
	\subfigure{\includegraphics[width=0.49\columnwidth]{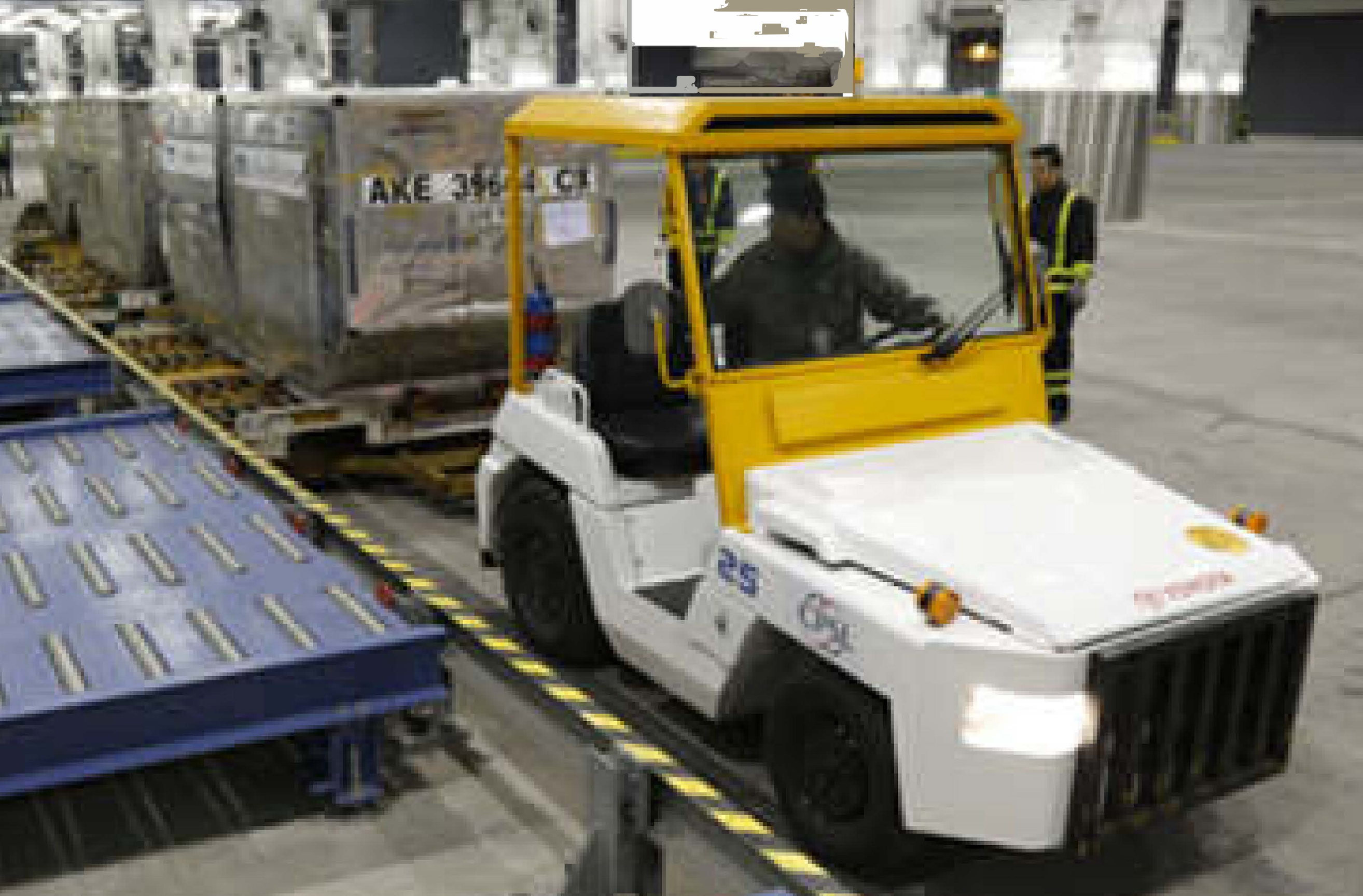}\label{fig:plane}}
	\subfigure{\includegraphics[width=0.49\columnwidth]{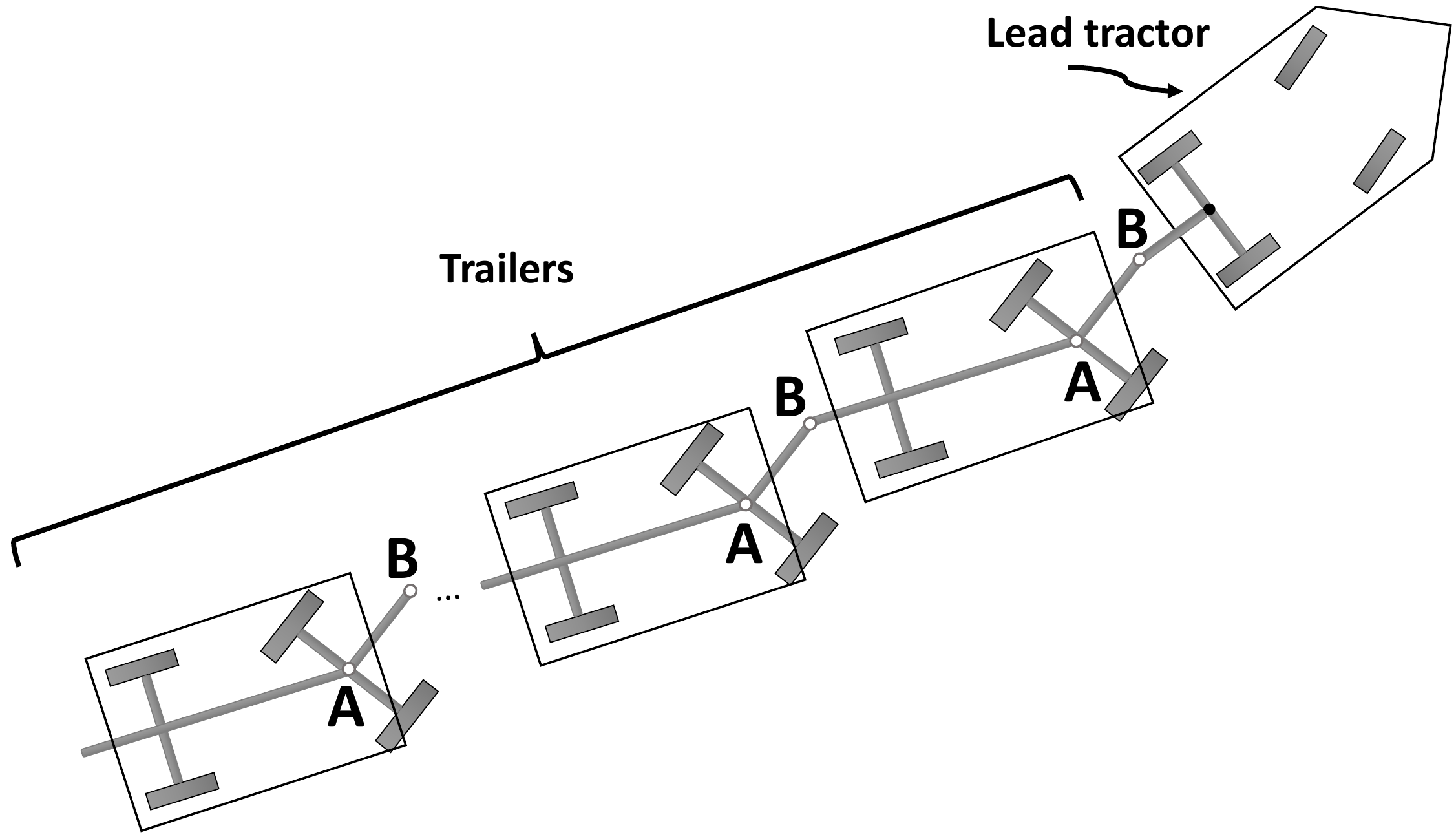}\label{fig:dist_it}}\\
	\vspace{-0.2cm}
	\caption{Tractor-trailers vehicle and its geometric model.}
	\label{realandgeomodel}
\end{figure}

In our work, we treat the on-axle hitching in the system as an off-axle hitching whose front connection bar has zero length. Partially using the ideas and methods from \cite{11}, trajectories that ensure safety handling of dynamic industrial scenarios can then be planned. Though the planned trajectory is for the lead tractor, it can simultaneously guarantee the whole system to be collision-free and executable for time-critical maneuvers. Therefore, the control issue of the tractor-trailers system can be simplified to the control problem of the single lead tractor. Since this paper focuses on trajectory control, details for trajectory planning will not be presented and we assume the reference trajectory is given in our problem.

To the best of our knowledge, this paper is the first to report trajectory control of full-size industrial tractor-trailers vehicles. The novelty of this work manifests in two aspects.

Firstly, a two-dimensional force sensor is installed at the connection (rightmost B point in Fig. \ref{realandgeomodel}) between the tractor and the first trailer. The tractor's dynamic model that explicitly accounts for the measured forces is introduced. A trajectory tracking controller is proposed to compensate the forces in real time and simultaneously drive the tractor onto the desired trajectory. In contrast to existing works, by reflecting trailers' motion effects onto the real-time force measurements, our approach considers the complex dynamic effects introduced by trailers, such as varying number and diverse configurations of and different payloads on the trailers. Moreover, since the lead tractor is the final object to be controlled, the controller does not require knowledge of trailers' dynamic parameters (mass, moment of inertia, etc.) or the unavailable states of the trailers.

Secondly, though the control of wheeled mobile robots has always been a hot topic for numerous research activities (see \cite{12},\cite{13},\cite{18},\cite{20},\cite{22}), real-world implementation of trajectory tracking for full-scale vehicles is barely seen, let alone for the tractor-trailers systems. The main challenge is to achieve high precision of acceleration and deacceleration commands. Hence, our second novelty arises from applying a fifth-order polynomial to model the map between propulsion force and throttle opening and velocity. The parameters are estimated by fitting the experimental data in order to provide accurate driving force for propulsion of the system and compensation of trailer dynamics, when directly controlling the throttle. 


The paper is organized as follows. Section II presents the dynamic model of the lead tractor. Section III details the controller design and stability analysis. Section IV shows the practical implementation on the full-size industrial tractor-trailers vehicle. Section V concludes the work.

\section{Dynamic Modeling}
The tractor is a rear-drive, front-steer car-like vehicle. To derive its dynamics, the nonlinear one-track model \cite{14} is employed where the front and rear wheels are respectively replaced by an intermediate wheel in the middle, see Fig. \ref{nonholonomicModel}. It is assumed that the motion is planar and the height of Center of Gravity (COG) is zero, thus the road bank and grade are not considered and the roll and pitch dynamics are neglected. The rolling resistance and aerodynamics drag are also assumed to be negligible. In this sense, the tractor is subject to the driving force $F_d$ that is modeled to act at center of the rear wheel, two lateral slip forces $F_r$, $F_f$ applied perpendicular to the wheels and two measurable forces $H_x$ and $H_y$ that are exerted on the tractor at the off-axle hitch point by trailing trailers. Be noted that the hitch joint does not transmit torque. Applying the Newton-Euler method, we obtain the following dynamic equations.
\begin{figure}[!t]
\centering
\includegraphics[width=0.9\columnwidth]{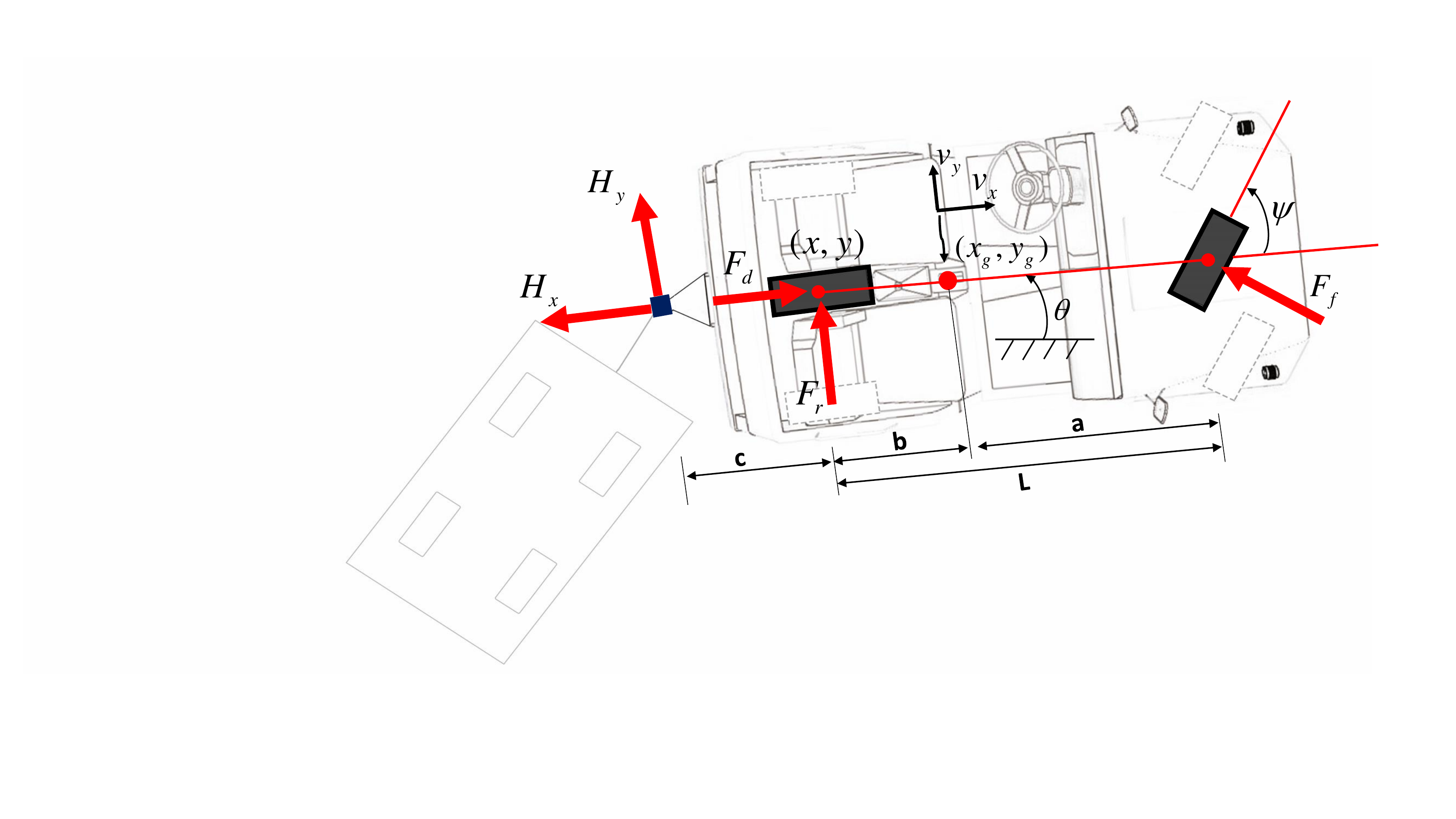}
\caption{Nonholonomic tractor model with force measurements.}
\label{nonholonomicModel}
\vspace{-0.5cm}
\end{figure}
\begin{align}
m\dot{v}_x & = F_d -F_f\sin\psi - H_x + mv_y\dot{\theta} \label{NE_dynamics1} \\
m\dot{v}_y & = F_r + F_f\cos\psi + H_y - mv_x\dot{\theta} \label{NE_dynamics2}\\
J\ddot{\theta} & = aF_f\cos\psi - bF_r - (b+c)H_y \label{NE_dynamics3},
\end{align}
where $m$ is the tractor mass, $J$ is the moment of inertia, $v_x$, $v_y$ are longitudinal and lateral velocity (in tractor's local frame) at COG. $\psi$ is the steering angle, $\theta$ is tractor's yaw angle, $a$, $b$ are distances from COG to front wheel and rear wheel respectively and $c$ is the distance between the force sensor and the rear wheel. The position of COG is denoted as $(x_g,y_g)$ in the inertial frame.
For our application, considering the required low speed of the tractor ($\leq8$ km/h when towing trailers in industrial scenarios), slippage-free condition can always be assumed to hold at the wheels \cite{14}. Hence, two nonholonomic constraints (\ref{nonholo_cstr1}), (\ref{nonholo_cstr2}) respectively for rear and front wheels can further be introduced.
\begin{align}
\dot{y}_g \cos \theta - \dot{x}_g \sin \theta  & =   b \dot{\theta}
\label{nonholo_cstr1} \\
\dot{y}_g \cos (\theta + \psi) - \dot{x}_g \sin (\theta + \psi)  & = - a \dot{\theta} \cos \psi
\label{nonholo_cstr2}
\end{align}
The rotational transformation between velocity in local and inertial frame is as below.
\begin{equation}\label{rt}
\begin{pmatrix} \dot{x}_g\\ \dot{y}_g\end{pmatrix}=
\begin{pmatrix} \cos\theta& -\sin\theta\\ \sin\theta& \cos \theta \end{pmatrix} \begin{pmatrix} v_x \\ v_y\end{pmatrix}
\end{equation}
Substituting (\ref{rt}) into (\ref{nonholo_cstr1}), (\ref{nonholo_cstr2}) yields (\ref{vel_rel}) and it's differentiation (\ref{vel_diff}).
\begin{equation}
v_y = b\dot{\theta}, \quad \dot{\theta}=\frac{v_x\tan\psi}{L}
\label{vel_rel}
\end{equation}
\begin{equation}
\dot{v}_y = b\ddot{\theta}, \quad \ddot{\theta}=\frac{\dot{v}_x\tan\psi}{L} + \frac{\dot{\psi}v_x}{L\cos^2\psi}
\label{vel_diff}
\end{equation}
Combining (\ref{NE_dynamics1}) - (\ref{vel_diff}) and choosing $\begin{bmatrix}x_g&y_g&\theta&v_x\end{bmatrix}$ as part of the states, we can get the simplified vehicle dynamic equations (\ref{final_dyn1}) - (\ref{final_dyn4}):
\begin{align}
\dot{x}_g & =v_x \begin{bmatrix}\cos \theta - \frac{b\sin\theta \tan\psi}{L}\end{bmatrix}\label{final_dyn1}
 \\
\dot{y}_g & =v_x \begin{bmatrix}
  \sin \theta + \frac{b\cos\theta \tan\psi}{L} \end{bmatrix}
\label{final_dyn2} \\
\dot{\theta} & = v_x \frac{\tan \psi}{L}    \label{final_dyn3} \\
\dot{v}_x & = \varphi_1 + \varphi_2 (F_d - H_x) - \varphi_3 H_y  \label{final_dyn4}
\end{align}
with
\begin{align}
\varphi_1 &= -(1/Z)(mb^2 + J)\tan\psi \dot{\psi} v_x  \nonumber
\\
\varphi_2 &= (1/Z)L^2\cos^2\psi  \nonumber \\
\varphi_3 &= (1/Z)L^2c \sin \psi \cos \psi \nonumber \\
Z  &= \cos^2\psi [L^2m + (mb^2 + J)\tan^2\psi ]
\end{align}
Furthermore, to ensure and improve the capability of trajectory tracking control, it is essential to model the actuation process. The dynamic behavior of the steering actuator is assumed, and experimentally proven, to be well-captured by a first-order lag element
\begin{equation}\label{steermodel}
\dot{\psi} = (u_2 - \psi) \frac{1}{\tau},
\end{equation}
where $\tau$ denotes the time constant and $u_2$ is the control input for steering.
The driving force $F_d$ is either the propulsion force $F_p$ transmitted to the rear wheel from the engine through driveline or braking force $F_b$ generated by brake actuation. The engine torque is normally modeled as being proportional to the throttle opening $u_1$ and also being a second-order polynomial w.r.t the engine speed \cite{15}, \cite{16}. The engine speed can also be proportionally related to vehicle speed $v_x$, when assuming that the torque converter in the tractor is locked and that there is no slippage at the wheels \cite{17}. The propulsion force can be further assumed to have a proportional relationship with engine torque \cite{15}. In light of these ideas and also considering the inherently complex nonlinear mapping of the tractor's driveline, we propose the fifth-order-polynomial propulsion map (\ref{engineMap}) between the propulsion force and the throttle opening and vehicle velocity. The coefficients $\boldsymbol{\Phi}$ are estimated by fitting experimental data. As shown in Fig. \ref{enginemap}, each of the black curves corresponds to a set of experimental force-speed data under a certain level of throttle opening. The throttle opening value is set from 0 to 300 and a few samples between are chosen to generate different black curves. The velocity is less than 5 $m/s$, which is the normal operation range for industrial tractors. The propulsion force is achieved in real time with an EKF force estimator combining IMU readings, encoders on the steering wheel and driven wheels and the vehicle dynamic model. The fifth-order polynomial is found to be the best candidate to fit the data as a propulsion map in the sense of simplicity and accuracy, compared to polynomials of other orders. The colored surface in Fig. \ref{enginemap} is the resultant polynomial map after data fitting. 
\begin{figure}[]
	\centering
	\includegraphics[width=0.8\columnwidth]{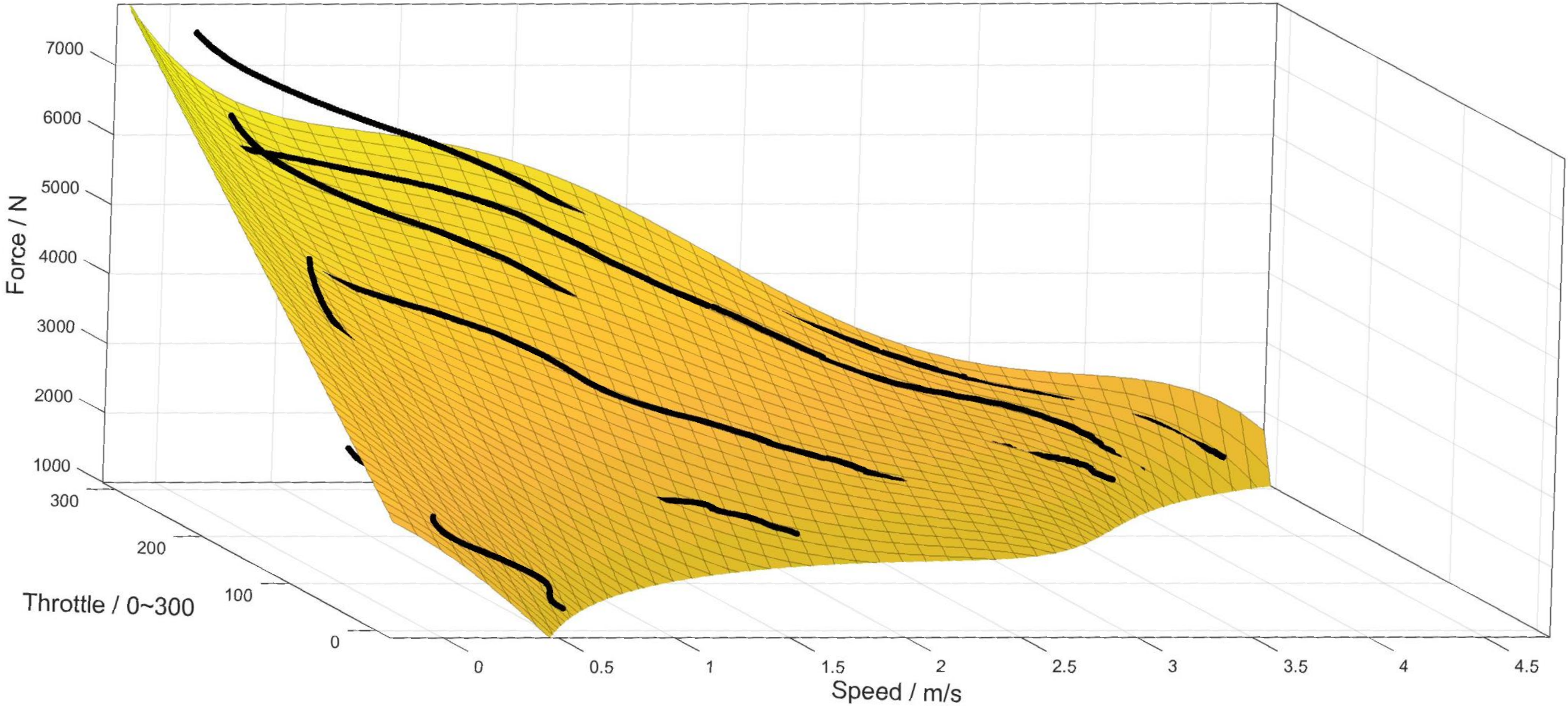}
	\caption{Map between propulsion force and throttle opening and velocity.}
	\label{enginemap}
	\vspace{-0.5cm}
\end{figure}

\begin{equation}
F_p = u_1\begin{bmatrix}1 & v_x & v_x^2 & v_x^3 & v_x^4 & v_x^5 \end{bmatrix} \begin{bmatrix} \beta_1 \\ \beta_2 \\ \beta_3 \\ \beta_4 \\ \beta_5 \\ \beta_6  \end{bmatrix} = u_1\boldsymbol{V}^T\boldsymbol{\Phi} \label{engineMap}
\end{equation}
The braking force $F_b$ generated by brake actuation is modeled to be proportional to the available control input of braking pressure $u_3$.
\begin{equation}\label{brakeModel}
	F_b = u_3 n_b
\end{equation}
with $n_b$ the constant ratio.
Finally, equations (\ref{final_dyn1}) - (\ref{brakeModel}) constitute the complete dynamic model of the tractor.

\section{Nonlinear Trajectory Control}

\subsection{Trajectory Representation}
Since the dynamic equations for heading (\ref{final_dyn3}), velocity (\ref{final_dyn4}), steering (\ref{steermodel}) and driving force (\ref{engineMap})(\ref{brakeModel}) do not depend on the position of the tractor, the coordinates of interest $(x_g,y_g)$ can be redefined as being the center point $(x,y)$ of the tractor's rear axle. Therefore, equations (\ref{final_dyn1}) and (\ref{final_dyn2}) can be rewritten as
\begin{align}
\dot{x} & = v_x \cos \theta\\
\dot{y} & = v_x \sin \theta \label{centerrealaxis}
\end{align}
The trajectory planning is then to continuously parameterize the desired position and orientation of the point $(x,y)$. As mentioned, the reference trajectory is planned for the tractor, and it also respects collision-free conditions for the whole tractor-trailers system. Besides, the planned trajectory is required to guarantee sufficient smoothness and account for actuation limitations introduced by the structure of the whole system. The reference trajectory is assumed to be given and is formulated as follows
\begin{align}
\boldsymbol{q}_d(t) & = [x_d(t),y_d(t),\theta_d(t),v_d(t)]^T\nonumber \\
\dot{\boldsymbol{q}}_d(t) & =[ \dot{x}_d(t),  \dot{y}_d(t) , \dot{\theta}_d(t) , \dot{v}_d(t)]^T
\end{align}
with
$$v_d = \sqrt{\dot{x}_d^2(t) + \dot{y}_d^2(t)}, \quad \theta_d = \atantwo(\dot{y}_d(t),\dot{x}_d(t)) \nonumber$$
It should be noted that, to avoid the jack-knife phenomenon in tractor-trailers system, the desired velocity $v_d$ is always positive.

\begin{figure}[!t]
\centering
\includegraphics[width=0.8\columnwidth]{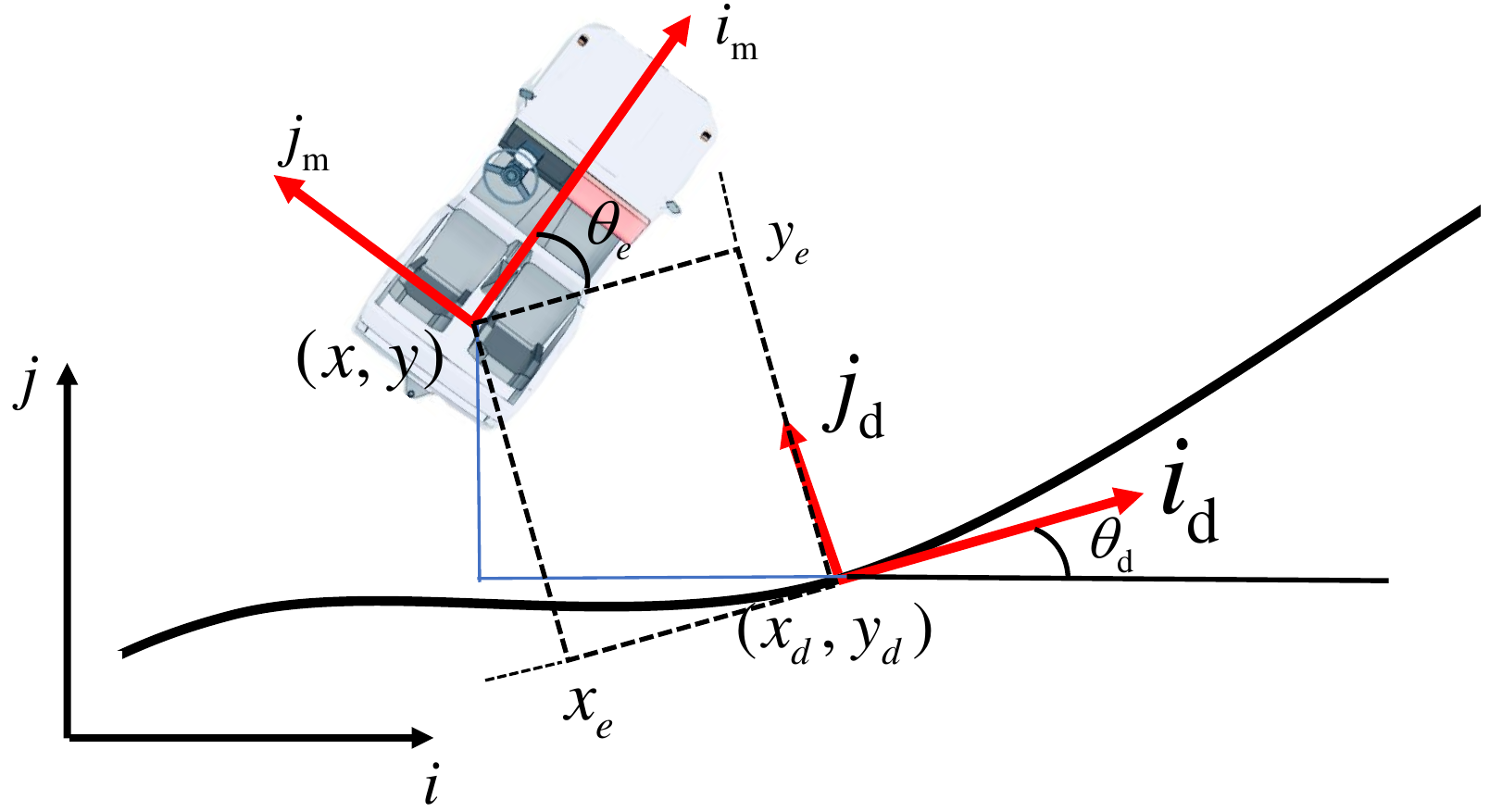}
\caption{Transformation of error coordinates.}
\label{transformation}
\vspace{-0.3cm}
\end{figure}

\subsection{Switching Policy between Throttle and Brake Control}
In our work, the control input for driving force is exclusively throttle opening rate $u_1$ or braking pressure $u_3$, which means only one type of input applies in a control cycle to generate $F_d$. The switch between throttle and brake depends on the sign of the desired value $\omega_2$ of the driving force, which is provided by the nonlinear control law presented in the following subsection. The switch policy is given by
\begin{equation}\label{switchPolicy}
 F_d = \left\{
\begin{array}{c l}	
u_1\boldsymbol{V}^T\boldsymbol{\Phi} &  \omega_2 \geqslant  0\\
u_3n_b & \omega_2 < 0
\end{array}\right.
\end{equation}
\subsection{Trajectory Tracking Control Design}
We first define the posture and velocity errors with respect to a Frenet frame as follows (see Fig. \ref{transformation}).
\begin{align}\label{diffeomorphicTrans}
x_e & = (x - x_d)\cos \theta_d + (y - y_d)\sin \theta_d \nonumber \\
y_e & = (y - y_d)\cos \theta_d - (x - x_d)\sin \theta_d \nonumber \\
\theta_e & = \theta - \theta_d  \nonumber \\
v_e & = v_x - v_d
\end{align}
Taking the derivative and applying equations (\ref{final_dyn3})-(\ref{centerrealaxis}), the system finally becomes
\begin{align}
	\dot{x}_e &= v_x \cos \theta_e + \dot{\theta}_d y_e -v_d  \label{errdyn_1}\\
	\dot{y}_e &= v_x \sin \theta_e - \dot{\theta}_d x_e \label{errdyn_2} \\
	\dot{\theta}_e &= v_x c_\psi - \dot{\theta}_d \label{errdyn_3} \\
	\dot{v}_e &= \varphi_1 + \varphi_2(F_d - H_x) - \varphi_3 H_y - \dot{v}_d \label{errdyn_4} \\
	\dot{c_\psi} &=  (\frac{\tan\psi}{L})' = (\frac{1}{L} + Lc_\psi^2)\dot{\psi}\label{errdyn_5} \\
	\dot{\psi} &= (u_2 - \psi) \frac{1}{\tau} \label{errdyn_6} \\
	 F_d &= \left\{
	\begin{array}{c l}	
	F_p&= u_1\boldsymbol{V}^T\boldsymbol{\Phi}, or \label{errdyn_7} \\
	F_b&= u_3n_b
	\end{array}\right.
\end{align}
Therefore, when all the states are available, the tracking control objective is to find control laws for throttle opening rate $u_1$, steering input $u_2$ and braking pressure $u_3$, such that $\lim_{t \to \infty }{[x_e(t), y_e(t), \theta_e(t), v_e(t)]} = \boldsymbol{0}$.

Consider, first, only the truncated system (\ref{errdyn_1}) - (\ref{errdyn_4}), and pretend that $c_\psi$ in (\ref{errdyn_3}) and $F_d$ in (\ref{errdyn_4}) can be directly manipulated by $\omega_1$ and $\omega_2$ respectively. We propose the following lemma.
\begin{lemma}
	The virtual control input $\omega_1$ and $\omega_2$ given by
	\begin{align}
	\omega_1  =& \frac{1}{\theta_e} [x_e (1 - \cos \theta_e) - y_e \sin\theta_e ] - k_{\theta}\theta_e + \frac{\dot{\theta}_d}{v_d} \label{con_w1}\\	
	\omega_2  =& H_x + \frac{1}{\varphi_2}(\varphi_3 H_y - \varphi_1 + \dot{v}_d - k_v v_e - x_e + k_\theta \theta_e^2 \nonumber \\
	&- \frac{\dot{\theta}_d}{v_d}\theta_e) \label{con_w2}
	\end{align}
    with $k_\theta,k_v > 0 $, makes $[x_e(t), y_e(t), \theta_e(t), v_e(t)] = \boldsymbol{0}$ of the partial system (\ref{errdyn_1}) - (\ref{errdyn_4}) globally asymptotically stable.
\end{lemma}
\begin{proof}
	The scalar function $V_1$ is proposed to be a Lyapunov function candidate,
	\begin{equation}
		V_1 = \frac{1}{2}(x_e^2 + y_e^2 + \theta_e^2 + v_e^2)
	\end{equation}
	its time derivative is given by
	\begin{equation}
		\begin{aligned}
		\dot{V}_1 =& x_e\dot{x}_e + y_e\dot{y}_e + \theta_e \dot{\theta}_e + v_e \dot{v}_e  \nonumber \\
		=& x_e(v_x \cos \theta_e + \dot{\theta}_d y_e -v_d) + y_e(v_x \sin \theta_e - \dot{\theta}_d x_e) \nonumber \\	
		& + \theta_e(v_x\omega_1 - \dot{\theta}_d) + v_e [\varphi_1 + \varphi_2(\omega_2 - H_x) \nonumber \\
		& - \varphi_3 H_y - \dot{v}_d] = -k_\theta v_d \theta_e^2 - k_v v_e^2 \leq 0
		\end{aligned}
	\end{equation}
Considering the boundness of variables inside the system, from Barbalat's lemma \cite{19}, we can easily show that $\lim_{t \to \infty } \dot{V_1} = 0$. Then,
$$ \lim_{t \to \infty } \theta_e = 0, \lim_{t \to \infty} v_e = 0$$
with $F_d = \omega_2$ and $v_e \equiv 0 \implies \dot{v}_e \equiv 0 $,  from (\ref{con_w2}) and (\ref{errdyn_4}), we get
\begin{equation}
	  x_e = 0 \label{xe_zero}
\end{equation}
with $c_\psi = \omega_1$ and $\theta_e \equiv 0 \implies \dot{\theta}_e \equiv 0 $, from (\ref{con_w1}) and (\ref{errdyn_3}), we get
\begin{equation}
y_e = 0 \label{ye_zero}
\end{equation}
Notice that in (\ref{con_w1}), terms $\sin\theta_e/\theta_e$ and $(1 - \cos\theta_e) / \theta_e$ have removable singularities and are accurately implemented with Taylor series approximation. Therefore, we can conclude the asymptotic stability of  $[x_e(t), y_e(t), \theta_e(t), v_e(t)] = \boldsymbol{0}$
\end{proof}

\subsubsection{Control Design for Throttle and Steering}

First consider the situation when the desired driving force $\omega_2 \geqslant 0 $, where $F_d = F_p = u_1\boldsymbol{V}^T\boldsymbol{\Phi}$.

Following the backstepping procedure \cite{19}, we further design the real control input $u_1$, $u_2$ such that $c_\psi$ and $F_p$ converge to the virtual control law $\omega_1$ and $\omega_2$ respectively. With the known polynomial coefficients $\boldsymbol{\Phi}$ and measurable velocity vector $\boldsymbol{V}$, driving force $F_d$ can be directly manipulated by $u_1$.  
Define the distance between $c_\psi$ and its goal $\omega_1$ as
\begin{equation}
	\delta_\psi = c_\psi - \omega_1 \label{delta_steering}
\end{equation}
we propose the following theorem,
\begin{theorem}
	When $\omega_2 \geqslant 0$, the control law for the throttle input $u_1$ and steering input $u_2$ given by
	\begin{align}
	u_1	&= \omega_2(\boldsymbol{V}^T \boldsymbol{\Phi})^{-1} \label{throttle_control}\\
	u_2 &= \frac{\tau}{\frac{1}{L} + Lc_\psi^2} (\dot{\omega}_1 - v_x \theta_e - k_\psi\delta_\psi ) + \psi
	\end{align}
	with $k_\psi > 0$, forces $[x_e,y_e,\theta_e,v_e,\delta_\psi]$ to asymptotically converge to zero.
\end{theorem}
\begin{proof}
	From (\ref{errdyn_5}) and (\ref{errdyn_6}), the time derivative of (\ref{delta_steering}) yields
	\begin{equation}
		\dot{\delta_\psi} = \dot{c}_\psi - \dot{\omega}_1 = (\frac{1}{L} + Lc_\psi^2) (u_2 - \psi) \frac{1}{\tau} - \dot{\omega}_1	\label{deltapsi_derivative}
	\end{equation}
   Define the new composite Lyapunov function candidate
   $$ V_2 = V_1 + \frac{1}{2} \delta_\psi^2 $$
   its time derivative yields
   \begin{align}
   \dot{V}_2 =& \dot{V}_1 + \delta_\psi \dot{\delta}_\psi  \\
             =& x_e(v_x \cos \theta_e + \dot{\theta}_d y_e -v_d) + y_e(v_x \sin \theta_e - \dot{\theta}_d x_e) \nonumber \\	
             & + \theta_e[v_x(\omega_1+\delta_\psi) - \dot{\theta}_d] + v_e [\varphi_1 + \varphi_2(u_1\boldsymbol{V}^T\boldsymbol{\Phi} - H_x) \nonumber\\
             &- \varphi_3 H_y - \dot{v}_d] + \delta_\psi(\dot{c}_\psi - \dot{\omega}_1 )     \nonumber \\
             =& -k_\theta v_d \theta_e^2 - k_v v_e^2 + \theta_e v_x \delta_\psi  \nonumber \\
             &+ \delta_\psi [(\frac{1}{L} + Lc_\psi^2) (u_2 - \psi) \frac{1}{\tau} - \dot{\omega}_1] \nonumber \\
             =& -k_\theta v_d \theta_e^2 - k_v v_e^2 - k_\psi \delta_\psi^2 \leq 0 \nonumber
   \end{align}
Analogous to Lemma 1, we can easily conclude that $\lim_{t \to \infty }\dot{V}_2 = 0$, which means
\begin{equation}
\lim_{t \to \infty } \theta_e = 0, \lim_{t \to \infty} v_e = 0,  \lim_{t \to \infty} \delta_\psi = 0
\end{equation}
With $\delta_\psi = 0$, and from (\ref{delta_steering}), we have $c_\psi = \omega_1$. With (\ref{throttle_control}), $F_p$ is directly controlled to be $\omega_2$ . Therefore, (\ref{xe_zero}) and (\ref{ye_zero}) still hold. Hence, the statement in the theorem can be readily proved.
\end{proof}

Be noted again that the derivatives of the expressions $\sin\theta_e/\theta_e$ and $(1-\cos\theta_e)/\theta_e$ in $\dot{\omega}_1$ also have removable singularities, which poses no problem for accurate implementation.

\subsubsection{Control Design for Brake and Steering}
When the desired driving force $\omega_2 < 0 $, then $ F_d = F_b = u_3 n_b $. We need to design steering input $u_2$ and braking input $u_3$, such that $c_\psi$ and $F_b$ converge to the virtual control input $\omega_1$ and $\omega_2$ respectively. With knowledge of the value $n_b$, $F_b$ can be directly manipulated by control input $u_3$. $u_2$ needs to be further designed.
We propose the following theorem,
\begin{theorem}
	When $\omega_2 < 0$, the control law for steering input $u_2$ and brake input $u_3$ given by
	\begin{align}	
	u_2 &= \frac{\tau}{\frac{1}{L} + Lc_\psi^2} (\dot{\omega}_1 - v_x \theta_e - k_\psi\delta_\psi  ) + \psi \\
	u_3	&= \omega_2/n_b
	\end{align}
	asymptotically stabilizes $ [x_e,y_e,\theta_e,v_e] = \boldsymbol{0} $ of the system (\ref{errdyn_1})-(\ref{errdyn_7}).
\end{theorem}
\begin{proof}
	Using again the lyapunov function candidate
   $$ V_2 = V_1 + \frac{1}{2} \delta_\psi^2 $$
	with (\ref{con_w1}), (\ref{con_w2}), (\ref{deltapsi_derivative}) and analogous to Lemma 1, we get
	\begin{equation}
	\dot{V}_2 = -k_\theta v_d \theta_e^2 - k_v v_e^2 - k_\psi \delta_\psi^2 \leq 0
	\end{equation}
	Mimicking the same arguments applied in the proof of Lemma 1 and Theorem 1, we can readily show that $ [x_e,y_e,\theta_e,v_e] = \boldsymbol{0} $ is asymptotically stable.
\end{proof}
\begin{figure}[!b] 	\label{experiment_fig}
	\centering
	\subfigure{\includegraphics[width=0.393\columnwidth]{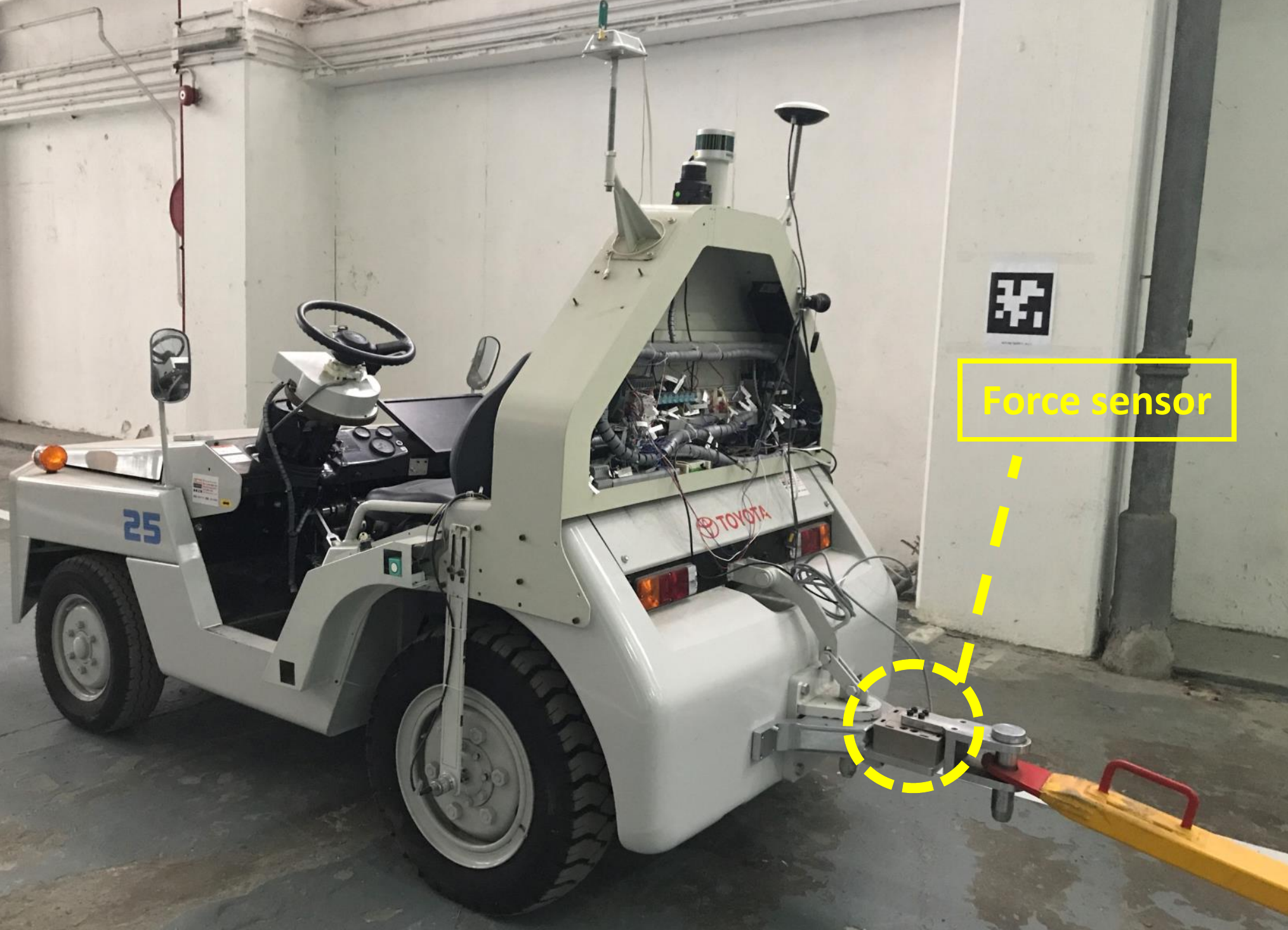}}
	\subfigure{\includegraphics[width=0.49\columnwidth]{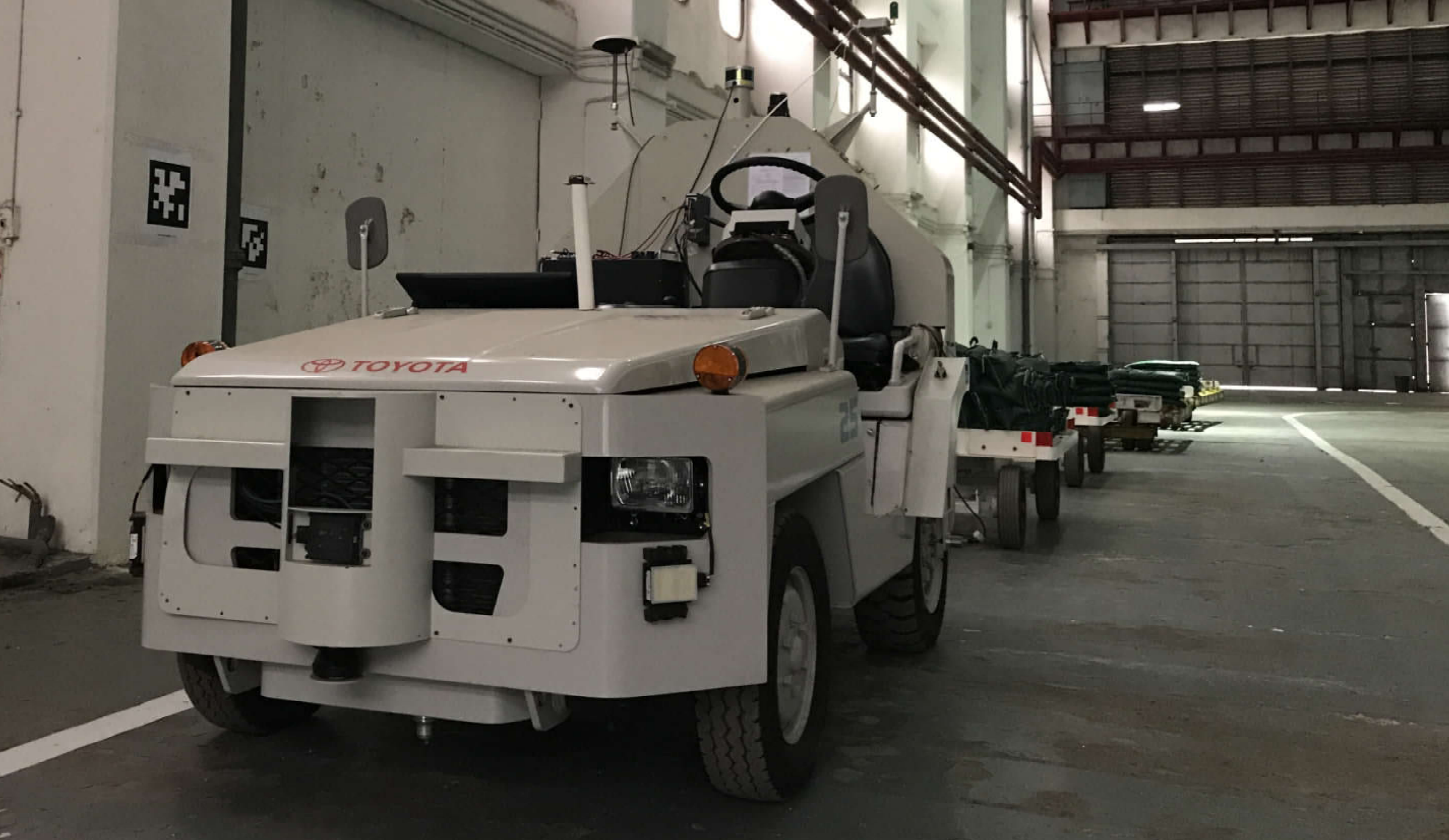}}\\
	\vspace{-0.2cm}
	\caption{The autonomous full-size industrial tractor-trailers vehicle.}
\end{figure}

\begin{figure}[]\label{experiment_1_fig}
	\centering
	\subfigure{\includegraphics[width=1\columnwidth]{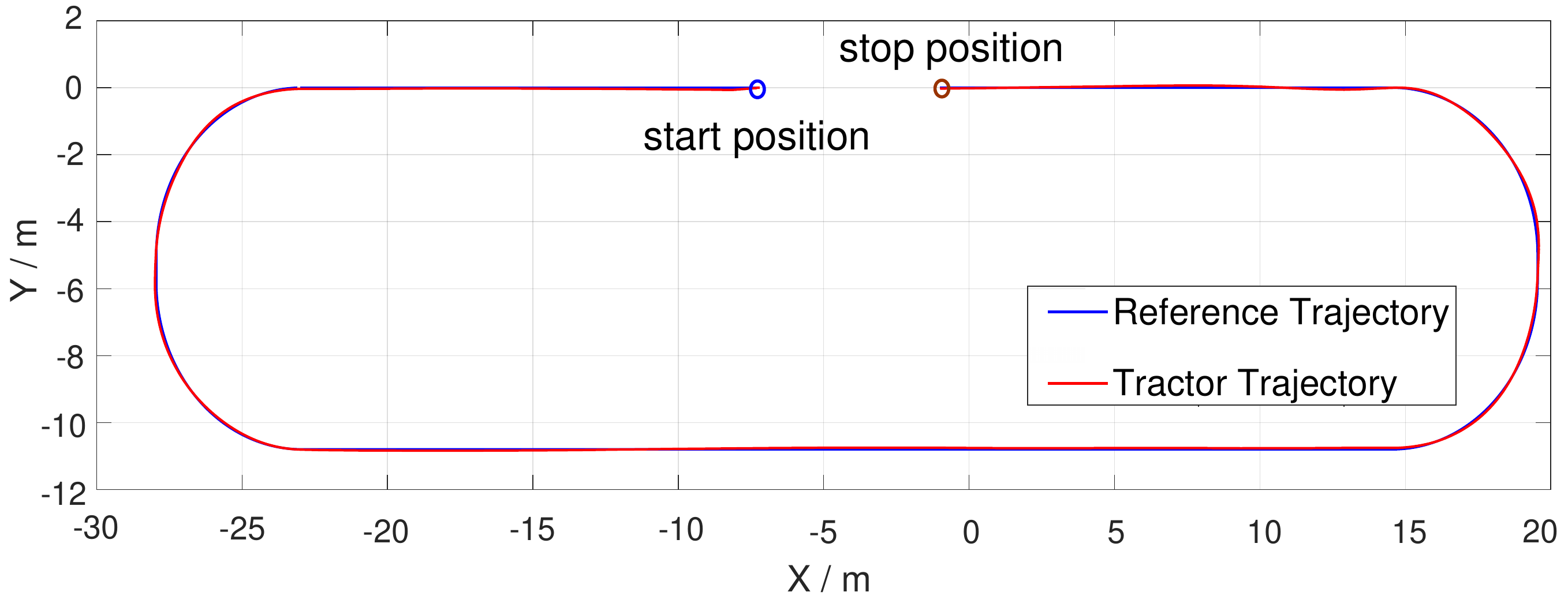}}
	\subfigure{\includegraphics[width=0.49\columnwidth]{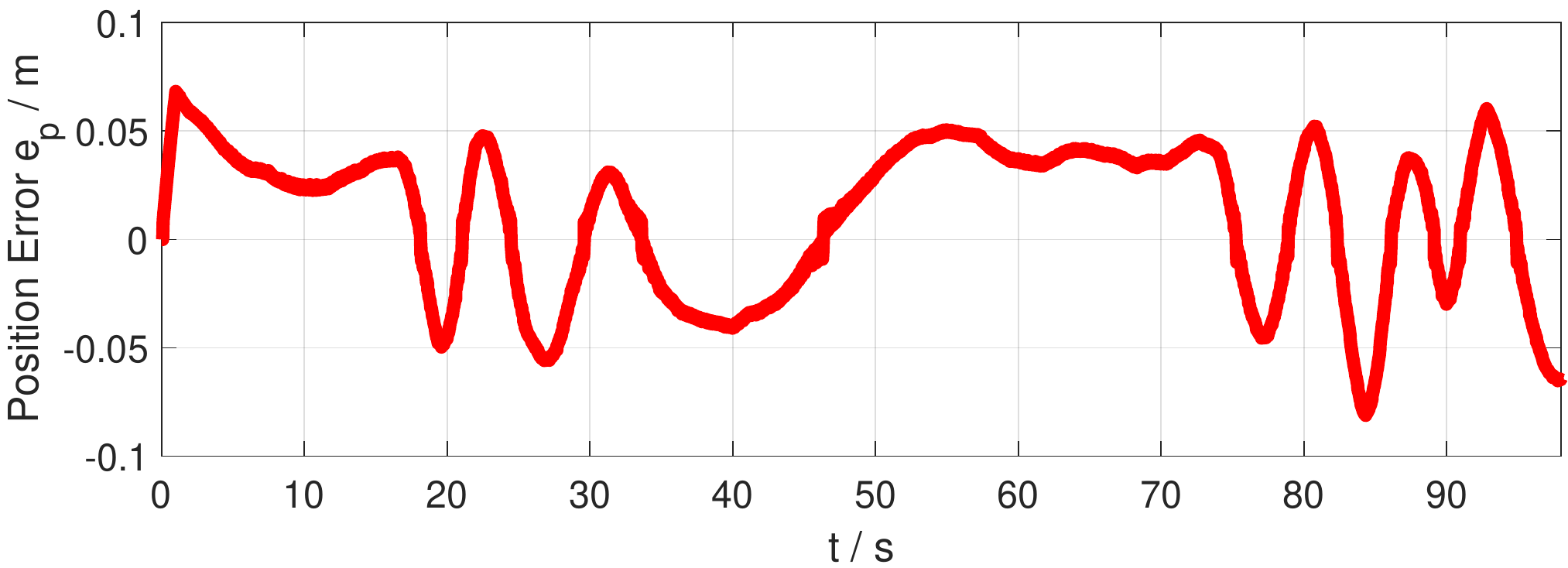}}
	\subfigure{\includegraphics[width=0.49\columnwidth]{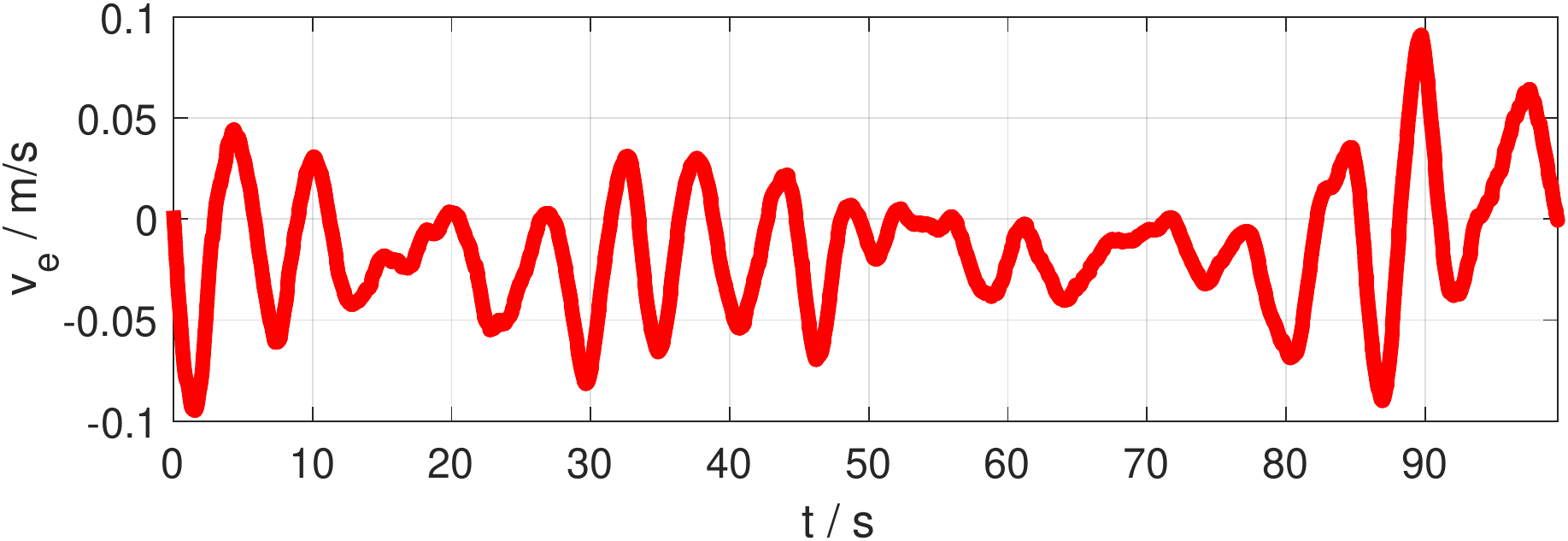}}\\
	\subfigure{\includegraphics[width=0.49\columnwidth]{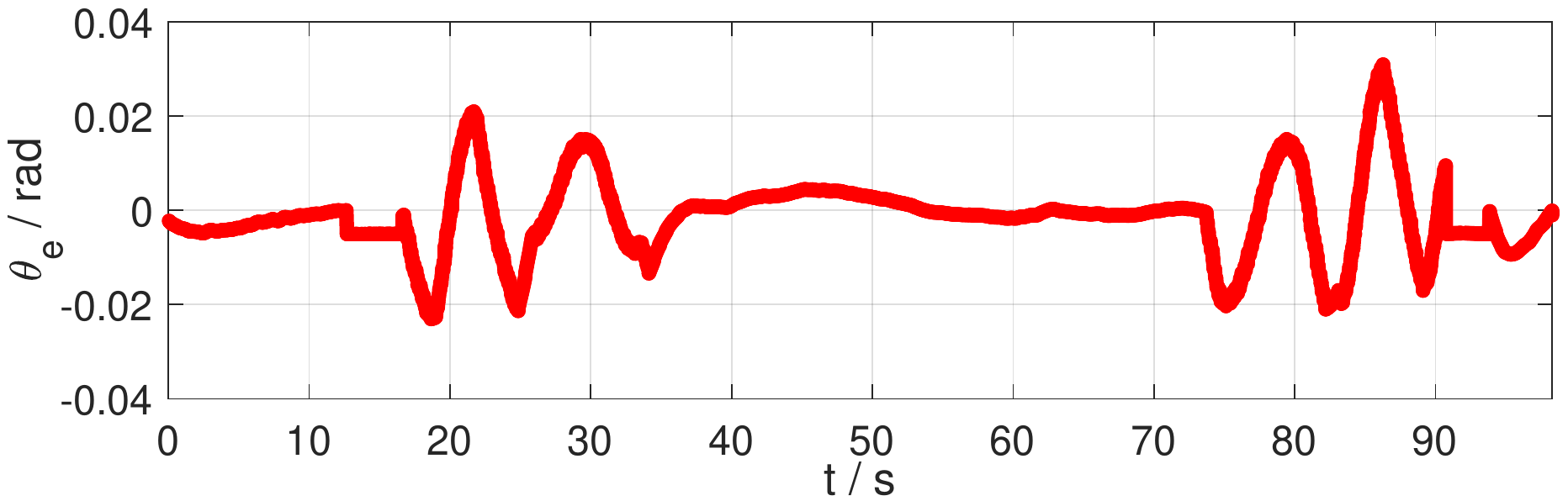}}
	\subfigure{\includegraphics[width=0.49\columnwidth]{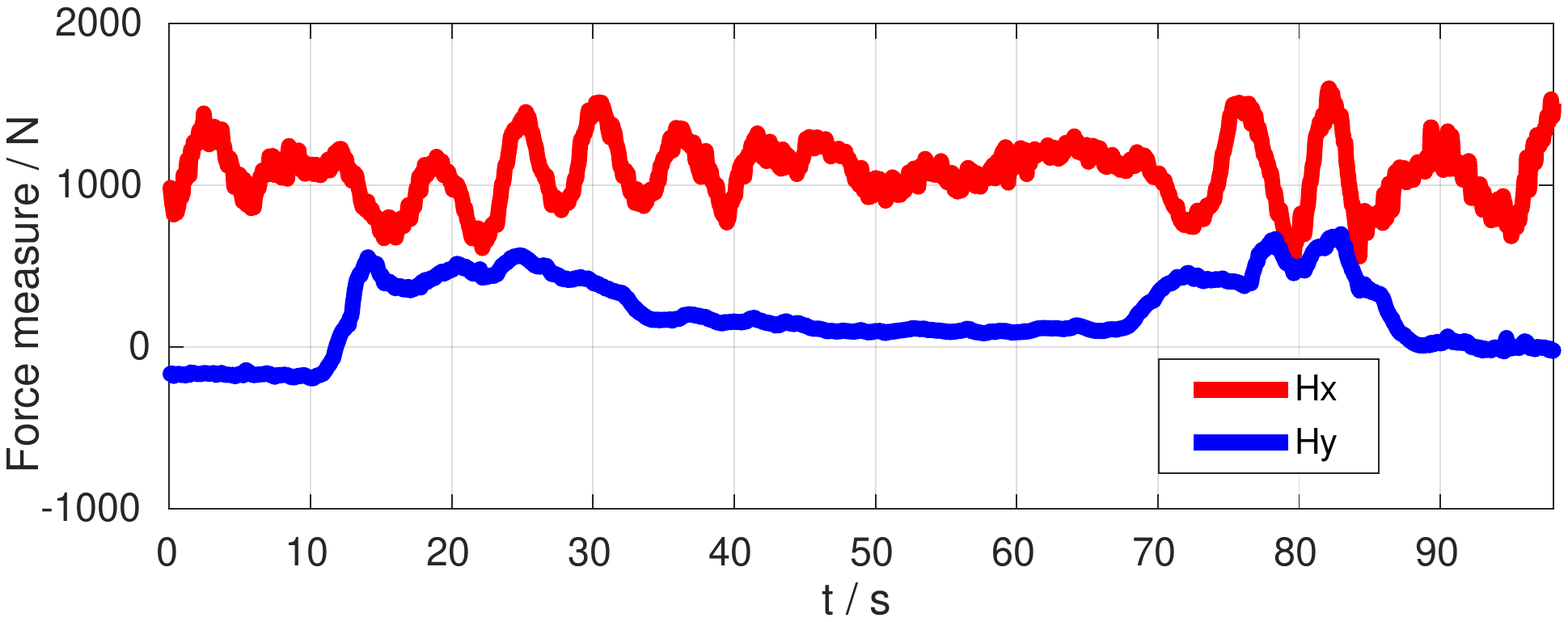}}\\
	\vspace{-0.2cm}
	\caption{Trajectory control performance with two full trailers and payload.}
\end{figure}

\begin{figure}[] 	\label{experiment_2_fig}
	\centering
	\subfigure{\includegraphics[width=1\columnwidth]{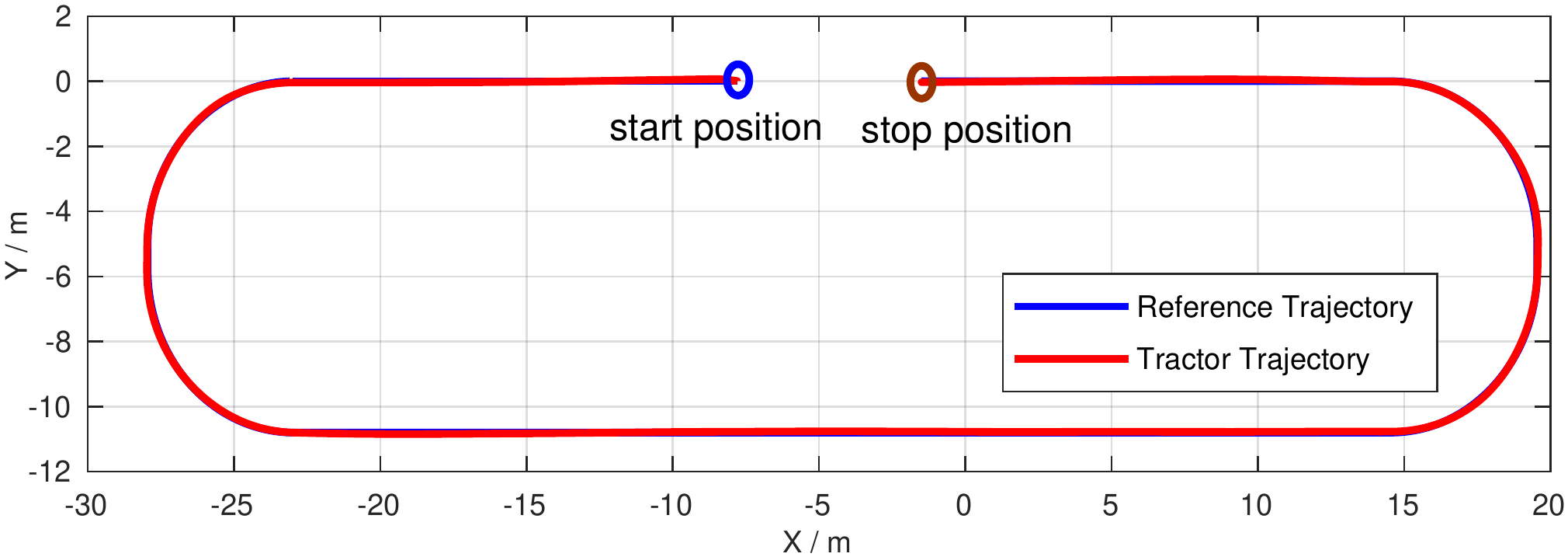}}
	\subfigure{\includegraphics[width=0.49\columnwidth]{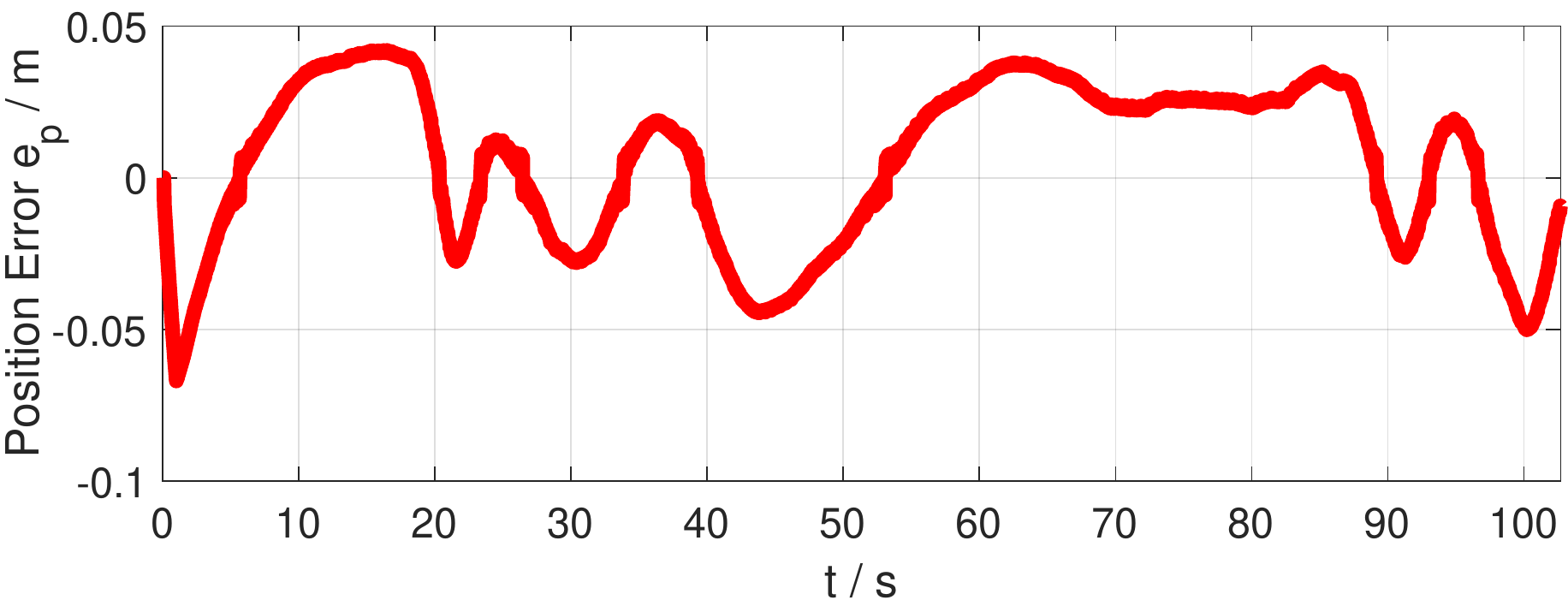}}
	\subfigure{\includegraphics[width=0.49\columnwidth]{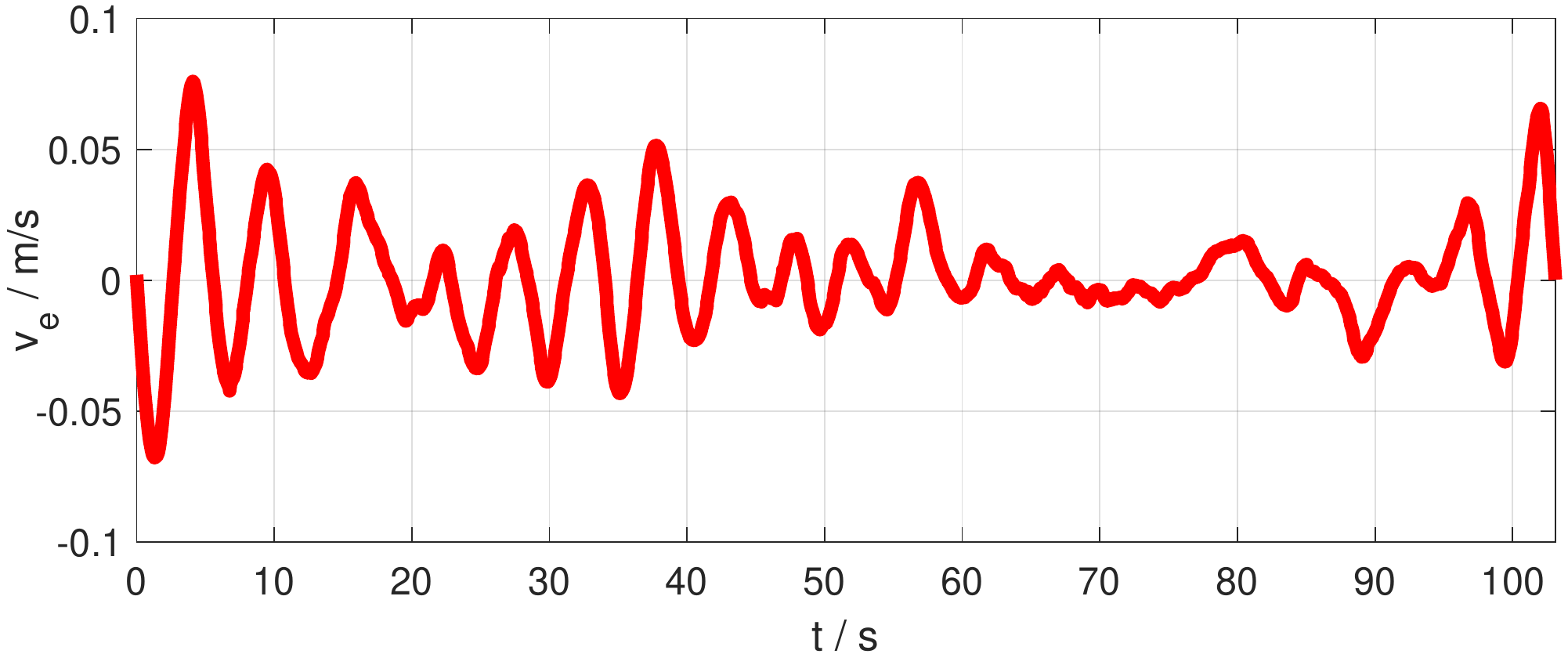}}\\
	\subfigure{\includegraphics[width=0.49\columnwidth]{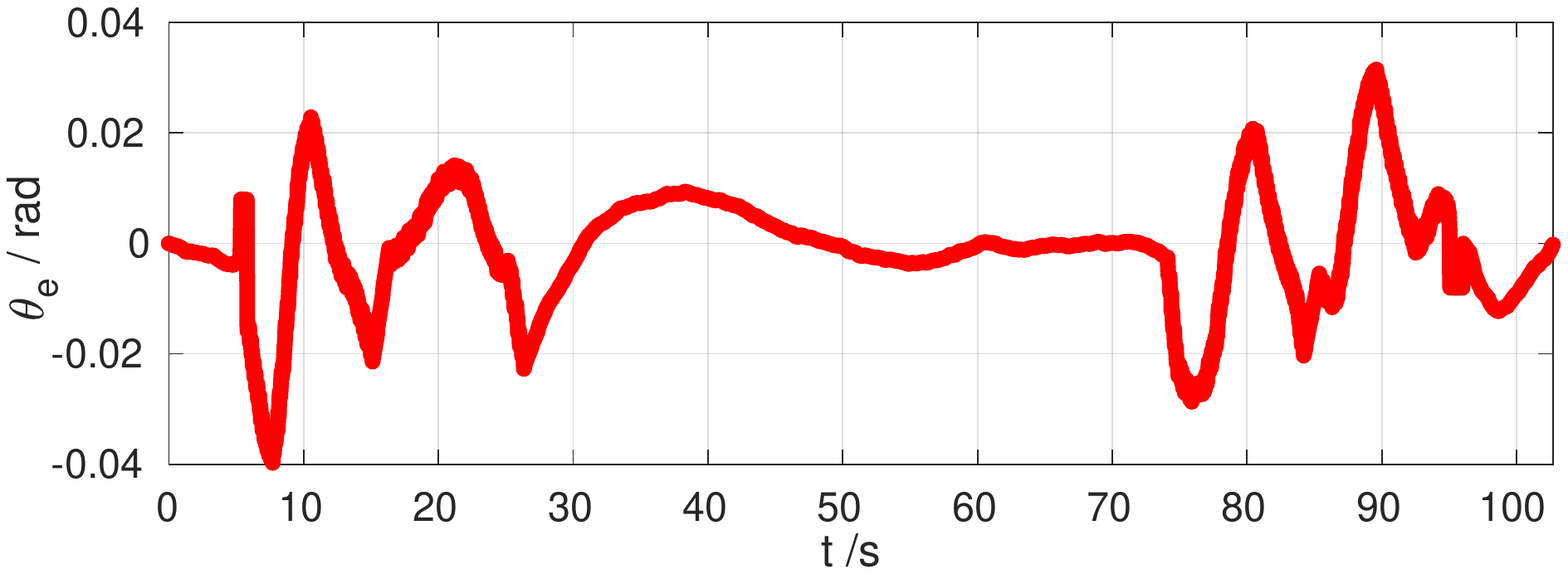}}
	\subfigure{\includegraphics[width=0.49\columnwidth]{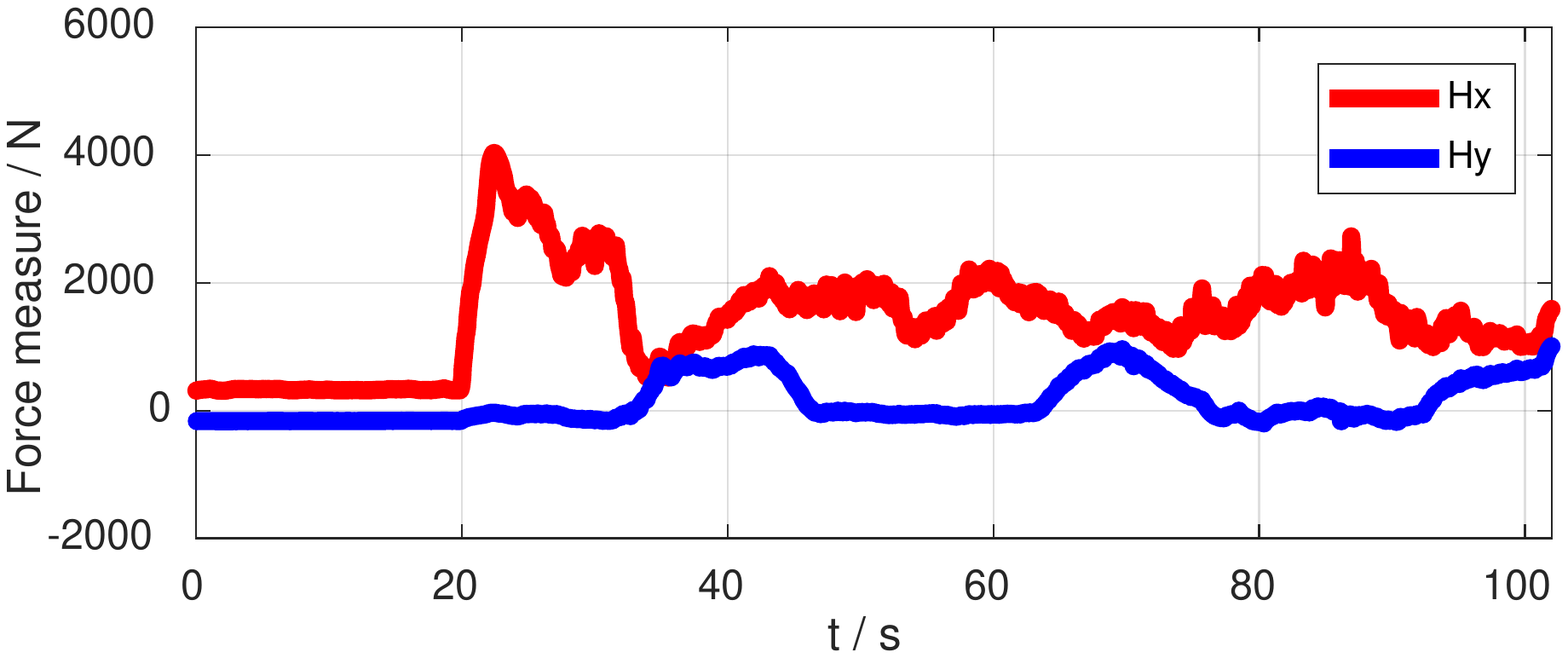}}\\
	\vspace{-0.2cm}
	\caption{Trajectory control performance with four full trailers and payload.}
\end{figure}

\subsubsection{Discussion}
 Since (\ref{diffeomorphicTrans}) is a diffeomorphic coordinate transformation, the convergence of error vector $[x_e,y_e,\theta_e,v_e]$ to zero guarantees the convergence of trajectory tracking.


\section{Practical Implementation and Results}
In this section, we implement the derived trajectory tracking controller with a full-size industrial tractor-trailers vehicle on an even test site inside a big warehouse (Fig. 5). The tractor is the \textit{Toyota 52-2TD25} model and is retrofitted with drive-by-wire throttling, braking and steering. The tractor as well as the full trailers (dollies) in the experiments are all standard models applied in the Hong Kong International Airport for luggage or cargo transportation. All the physical parameters of the tractor (mass, inertia, COG and dimensions) can be determined with standard measurements.

We adopt the 3-axis force sensor \textit{ME K3D160}  with a measuring range of $\pm50$KN (enough for normal industrial task execution) in each mutually perpendicular axes. Only two forces parallel to the road surface are counted. The force sensor is fixed on the tractor and is connected to the trailer bar via an intermediate link. Due to the small dimension of the link, we can assume that the trailer is directly connected to the force sensor with a hitch joint.
Yaw rate and acceleration measurements are obtained with \textit{Xsens MTi-300-AHRS} IMU. The self-localization process (providing position and orientation) is conducted by fusing sensory information from UWB, 3D laser scanner, cameras and inertial and odometry measurements. The coefficient vector is calculated off-line as $\boldsymbol{\Phi} = \begin{bmatrix}26.2 & -9.999 & 3.018 & -1.041 & 0.2354 & -0.021 \end{bmatrix} ^T$. 

We remove the obstacles and make the tractor running in a roundabout way to eliminate possible negative effects of the trajectory planning level on the tracking results. 
We conducted two experiments with different number and different size of trailers. 
Fig. 6 shows the results of the experiment where the tractor tows two relatively small full trailers. The weight of one small full trailer is 630kg. Besides, sandbags with total weight of 2000kg are loaded on the first trailer to act as extra payload. The desired trajectory is planned online and keeps the tractor running with a velocity of $1$ m/s along the trajectory. The total running time is $98$s. The position error is calculated by $\vert e_p(t)\vert = \sqrt{x_e^2(t) + y_e^2(t)}$, with its sign being negative when the tractor is on the left of the desired trajectory. Fig. 7 shows the results of the second experiment where the tractor tows two extra big full trailers, with each weighing 1300kg. The running time is $103$s and the constant desired velocity is 0.8 m/s. 
Above results demonstrate precise and robust performance of the proposed trajectory controller under different physical properties and dynamic configurations of the industrial tractor-trailers vehicle.

\section{Conclusion}
In order to handle the complex dynamics of industrial vehicles with tractor-trailer structure and finally realize the autonomous trajectory tracking control, a force sensor is proposed to be installed at the connection between the tractor and the first trailer. The tractor's dynamic model that explicitly accounts for the measured forces has been derived. Both  throttle \& steering control and brake \& steering control are proposed. The practical implementation results on full-size industrial tractor-trailers vehicles have shown good performance of the proposed approaches in trajectory tracking and the handling of complex dynamics. 

\bibliography{IEEEabrv,IROS2019_ref}

\end{document}